\documentclass[journal]{IEEEtran}

\newcommand{\be}{\begin{equation}}
\newcommand{\ee}{\end{equation}}
\newcommand{\bea}{\begin{eqnarray}}
\newcommand{\eea}{\end{eqnarray}}
\newcommand{\df}{{\rm d}}

\ifCLASSINFOpdf
\else
\fi
\usepackage{graphicx}

\usepackage{amsfonts}

\usepackage{color}

\usepackage{amsthm}

\newtheorem{Th}{Theorem}

\newtheorem{Co}{Corollary}

\newtheorem{Le}{Lemma}

\begin{document}
%
\title{A necessary and sufficient condition for minimum phase and implications for phase retrieval}
%
%
%

\author{Antonio Mecozzi,~\IEEEmembership{Fellow,~IEEE}
\thanks{Antonio Mecozzi is with the Department
of Department of Physical and Chemical Sciences, University of L'Aquila, 67100 L'Aquila, Italy e-mail: antonio.mecozzi@univaq.it.}
\thanks{Manuscript received April 19, 2005; revised September 17, 2014.}}

%
%

\markboth{IEEE TRANSACTIONS ON INFORMATION THEORY,~Vol.~13, No.~9, September~2014}%
{Shell \MakeLowercase{\textit{et al.}}: Bare Demo of IEEEtran.cls for Journals}
%



\maketitle

\begin{abstract}
We give a necessary and sufficient condition for a band-limited function $E(t)$ being of minimum phase, and hence for its phase being univocally determined by its intensity $|E(t)|^2$. This condition is based on the knowledge of $E(t)$ alone and not of its analytic continuation in the complex plane, thus greatly simplifying its practical applicability. We apply these results to find the class of all band-limited signals that correspond to distinct receiver states when the detector is sensitive to the field intensity only and insensitive to the field phase, and discuss the performance of a recently proposed transmission scheme able to linearly detect all distinguishable states.

\end{abstract}


%
\IEEEpeerreviewmaketitle

\section{Introduction}
%
%
%
%

\IEEEPARstart{P}{hase} retrieval is a longstanding problem in many fields of physics and applied sciences \cite{Fienup,Burge,Taylor,Shechtman}. Sufficient conditions ensuring that the phase of the signal and hence the full $E(t)$ can be reconstructed from the knowledge of the intensity profile only $|E(t)|^2$ are well known \cite{Burge,Taylor,Cassioli,Mecozzi}. A necessary and sufficient condition is also well known, and it is based on the position in the complex plane of the zeros of $E(z)$, the analytic continuation of $E(t)$. Such a condition however is of limited practical use, because the analytic continuation of a function is an ill posed problem and hence far from being amenable to simple numerical solutions \cite{Gallicchio}. Even when analytic continuation is possible, like for instance for bandwidth limited $E(t)$, finding the zeros in the complex plane is a difficult numerical task. Being instead the numerical evaluation of $E(t)$ for $t$ real relatively trivial, a necessary and sufficient condition based on the properties of $E(t)$ for $t$ real only would be a very useful tool.  To the best of the author's knowledge, however, such a condition has not been reported yet. The purpose of this paper is to derive such a condition, which we will see is very similar to the well-established Nyqvist stability criterion \cite{Bode} of control theory. We will consider for simplicity to band-limited signals only, and leave generalizations to wider classes of signals to future studies. 

The outline of this paper is the following. After presenting the derivation of a necessary and sufficient condition ensuring that the phase of a band-limited field can be retrieved from its intensity profile, we apply this condition to find the class of band-limited signals corresponding to distinct states when the receiver is insensitive to the signal phase and sensitive to the field intensity only. The implications for the capacity of an optical system using square-law detection at the receiver are finally discussed.


While emphasis will be given to examples belonging to the field of telecommunications only, the new condition discussed here can be of help in all fields of physics and applied sciences where the problem of phase reconstruction from the intensity only is an important one \cite{Shechtman}.

\section{A necessary and sufficient condition for minimum phase}

In this paper, we define the Fourier transform of functions $F(t) \in L^1 \bigcap L^2$ as
\be \tilde F(\omega) = \int_{-\infty}^\infty \df t \exp(i \omega t) F(t). \label{10} \ee
The fact that $F(t) \in L^1$ insures that $F(\omega)$ is uniformly continuous for $\omega \in \mathbb R$ \cite{Katznelson}.

\begin{Le} \label{Le1} Assume a function $E_s(t) \in L^1 \bigcap L^2$, such that its Fourier transform $\tilde E(\omega)$ is $\tilde E_s(\omega) = 0$,  $\forall \omega < 0$. Then we have
\be E_s(t) = \frac i  \pi \, \mathrm{p.v.}\int_{-\infty}^\infty \frac{E_s(t')}{t'-t} \df t' , \label{kk} \ee
where with p.v. we refer to the Cauchy's principal value of the integral. \end{Le} 


\begin{proof} Condition $\tilde E_s(\omega) = 0$,  $\forall \omega < 0$ implies that $\tilde E_s(\omega) = u(\omega) \tilde E_s(\omega)$, where $u(\omega)$ is a Heaviside unit step function. If we write
\be \tilde E_s(\omega) = \lim_{\epsilon \to 0^+} u(\omega) \exp(-\epsilon \omega) \tilde E_s(\omega), \ee
inverse Fourier transformation gives 
\be E_s(t) = \frac{i}{2 \pi} \int_{-\infty}^\infty \lim_{\epsilon \to 0^+} \frac{1}{t'-t+ i \epsilon} E_s(t') \df t', \ee
that is
\be E_s(t) = \frac{i}{2 \pi} \int_{-\infty}^\infty \lim_{\epsilon \to 0^+} \frac{t'-t - i \epsilon}{(t'-t)^2 + \epsilon^2} E_s(t') \df t'. \ee
Being $\lim_{\epsilon \to 0^+} \epsilon/[(t'-t)^2 + \epsilon^2] = \pi \delta(t'-t)$ where $\delta(\cdot)$ is the Dirac delta distribution, we obtain
\be E_s(t) = \frac{i}{\pi} \int_{-\infty}^\infty \lim_{\epsilon \to 0^+} \frac{t'-t}{(t'-t)^2 + \epsilon^2} E_s(t') \df t', \ee
that is relation (\ref{kk}). \end{proof}

Consequence of Eq. (\ref{kk}) is that the real and imaginary parts of $E_s(t) = E_{s,r}(t) + i E_{s,i}(t)$ are the Hilbert transform of one another
\bea E_{s,i} (t) &=& \frac 1  \pi \, \mathrm{p.v.}\int_{-\infty}^\infty \df t' \frac{E_{s,r}(t')}{t'-t}, \\
E_{s,r}(t) &=& - \frac 1  \pi \, \mathrm{p.v.}\int_{-\infty}^\infty \df t' \frac{E_{s,i}(t')}{t'-t}. \eea
These relations are known in spectroscopy as Kramers Kronig relations \cite{Kramers,Kronig}.

Let $\beta$ be a strictly positive constant with $0<\beta \le 1$ and let us define $B = (1+\beta)/T$. Let $\tilde \mathcal{C}_\beta (0,B)$ be the class of functions  $E_{s}(t)$ of the form
\be E_s(t) = \sum_{n=-\infty}^\infty  a_n \exp\left[- i (1+\beta) \frac{\pi (t-nT) }{T} \right] H_\beta (t-nT), \label{mod1} \ee
where we assume that the sequence of $a_n \in \mathbb C$ has a compact support, namely for any $E_s(t) \in \tilde \mathcal{C}_\beta (0,B)$, exists an $N \in \mathbb Z$ such that $a_n = 0$ for $|n| > N$. The orthogonal set of functions $\{H_\beta(t-nT), n \in \mathbb Z\}$ are defined as 
\be H_\beta(t) = \frac 1 {\sqrt{T}} \frac{\sin\left[\pi \frac t T (1-\beta)\right] + 4 \beta \frac t T \cos\left[\pi \frac t T (1+\beta)\right]}{\pi \frac t T \left[1-\left(4 \beta \frac t T \right)^2\right]}. \label{Hb} \ee
The functions in $\tilde \mathcal{C}_\beta (0,B)$ belong to $L^1 \bigcap L^2$. Their Fourier transform is
\be \tilde E_{s}(\omega) = \sum_{n=-N}^N a_n \exp\left(i n T \omega \right) \tilde H_\beta\left[\omega - (1+\beta) \pi /T \right], \label{Ew1} \ee
where
\be \tilde H_\beta(\omega) = \sqrt T \hspace{0.1cm} C_\beta \left(\frac{\omega T} {2 \pi}\right), \ee
with
\be C_\beta(x) = \left\{
\begin{array}{ll} 
1, & |x| \le \frac{1-\beta} 2; \\
\cos\left[ \frac{\pi}{2 \beta} \left(|x|-\frac{1-\beta} 2 \right) \right], & \frac{1-\beta} 2 < |x| \le \frac{1+\beta} 2; \\
0, & |x| > \frac{1+\beta} 2. \end{array} \right.
\ee
The functions $\{H_\beta(t-nT), n \in \mathbb Z\}$ have a square-root raised cosine spectrum, and their orthogonality can be directly verified in the Fourier domain
\be \int_{-\infty}^\infty |H_\beta(\omega)|^2 \exp\left[i (n'-n) T \omega \right] \frac{\df \omega}{2 \pi}= \delta_{n,n'}, \label{ort} \ee
where the integral is facilitated by noting that $\sum_{n} |H_\beta(\omega- 2 \pi T n)|^2 = T$ and using the periodicity of $\exp(i n T \omega)$.  Being $E_s(t) \in L^1 \bigcap L^2$, its spectrum is a continuous function of $\omega$. In addition, it is zero for $\omega \notin (0, 2 \pi B )$. Finally, it is easy to verify that $E_s(z)$, the analytic continuation of $E_s(t)$ in the complex plane, is an entire function. 

If we define $\tilde \mathcal{C} (0,B) = \lim_{\beta \to 0^+} \tilde \mathcal{C}_\beta (0,B)$ , any function $E_s(t) \in \tilde \mathcal{C}(0,B)$ can be written in the form
\be E_{s}(t) = \hspace{-0.2cm} \sum_{n=-N}^N a_n \exp\left[- i \frac{\pi (t-nT) }{T} \right] H_0(t-nT) \label{mod} \ee
where
\be H_0(t) =  \lim_{\beta \to 0} H_\beta(t) = \frac{1}{\sqrt{T}} \hspace{0.1cm} \mathrm{sinc}\left(\frac{\pi t} T\right). \ee
The Fourier transform of (\ref{mod}) is
\be \tilde E_{s}(\omega) = \sum_{n=-\infty}^\infty a_n \exp\left(i n T \omega \right) \tilde H_0\left(\omega - \pi /T \right), \label{Ew} \ee
where
\be \tilde H_0(\omega) = \sqrt T \hspace{0.1cm} C_0\left(\frac{\omega T} {2 \pi}\right), \ee
and $C_0(x) = \lim_{\beta \to 0} C_\beta(x)$ is a function equal to 1 in the open interval $(-1/2, 1/2)$, equal to $1/\sqrt 2$ for $x=\pm 1/2$ and zero elsewhere. Being for $\omega, \omega' \in (0, 2\pi B)$
\be \sum_n  \tilde H_0(\omega)  \tilde H_0(\omega') \exp\left[-i n T (\omega-\omega') \right] \nonumber \\
= 2 \pi \delta\left(\omega-\omega' \right), \ee
any function $E_s(t) \in L^2$ band-limited to the interval $(0,2 \pi B)$ can be expressed in the form (\ref{mod}), where the $a_n$ are given by
\be a_n = \int \tilde H_0(\omega') \exp\left(-i n \frac{2 \pi}{T}\right) \tilde E_s(\omega') \frac{\df \omega'}{2 \pi}. \ee
Although the functions of the form (\ref{mod}) belonging to $\tilde \mathcal{C}(0,B)$ are in $L^2$, they are in general not in $L^1$. For this reason, in the following where we refer to the class of band-limited functions $\tilde \mathcal{C}(0,B)$, we will assume that its components are members of the class $\tilde \mathcal{C}_\beta(0,B)$ for $\beta >0$, hence always in the class $L^1 \bigcap L^2$, approaching arbitrarily close the limit $\beta = 0$. 

\begin{Th} \label{Th1} Let us assume $E_s(t) \in \tilde \mathcal{C}_\beta (0,B)$, and a constant $\bar E \ne 0$, and define $E(t) = E_s(t) + \bar E$ with $\bar E \ne 0$, such that $E(t) \ne 0$, $\forall t \in \mathbb{R}$ and the trajectory of $E(t)$ never encircles the origin for $t \in (-\infty, \infty)$. Then, the number of zeros of $E(t+i\tau)$ with $\tau <0$ is equal to the winding number (i.e, the number of windings) around the origin of the curve described by $E(t)$ when $t$ runs over the entire real axis from $t \to -\infty$ to $t \to \infty$. \end{Th}

\begin{proof}
The function $E(z)$ is an entire function. Let $\Gamma$ be a contour encompassing the lower complex plane $z = t + i \tau$, which incorporates the real axis from $t = -\infty $ to $\infty$ and returns to $-\infty$ from the lower complex half-plane by a semicircle $C$ with radius $\rho \to \infty$.  By the Cauchy's argument principle, we have
\be I_\Gamma = \frac{1}{2 \pi i} \oint_{\Gamma} \frac{\dot{E}(z)}{E(z)} \df z = N_\mathrm{zeros} - N_\mathrm{poles}, \label{argument} \ee
where $\dot{E}(t) = \df E(t)/\df t$, $N_\mathrm{zeros}$ is the number of zeros and $N_\mathrm{poles}$ is the number of poles of $E(z)$ encircled by $\Gamma$.  The function $E(t + i \tau)$ does not have poles for $\tau <0$, $N_\mathrm{poles}=0$.  Using the substitution $E(z) = v$, we obtain
\be I_\Gamma = \frac{1}{2 \pi i} \int_{v_\Gamma} \frac{\df v}{v} = N_\mathrm{zeros}. \ee
Let us now split the integral $I_\Gamma$ in (\ref{argument}) into two parts, $I_\Gamma = I_r+I_C$, the first
\be I_r = \frac{1}{2 \pi i} \int_{v_r} \frac{\df v}{v} \ee
being the contribution of $E(z)$ when $z$ belongs to the real axis, the second
\be I_C = \frac{1}{2 \pi i} \int_{v_C} \frac{\df v}{v}, \ee
being the contribution of $E(z)$ when $z$ belongs to the infinite semicircle $C$. The integration path $v_C$ is made of the trajectory described by $E(z)$ for $z = R \exp(i \phi)$ when $\phi \in [\pi, 2 \pi]$.  Being $E(z) \to \overline E$ on the infinite semicircle, the length of the integration path $v_C$ tends to zero. In addition, being $\bar E \ne 0$, the modulus of the integrand tends to a finite constant on the infinite semicircle, $|1/v| \to 1/|\bar{E}|$, so that $I_C = 0$. Consequently, only the first term survives,
\be I_r = \frac{1}{2 \pi i} \int_{v_r} \frac{\df v}{v} = N_\mathrm{zeros}. \label{IR} \ee
The curve $v_r$ is made of the trajectory described by $E(t)$ when $t$ runs over the entire time axis. Such curve is closed because $E(t)$ tends to $\overline E$ for both $t \to \infty$ and $t \to -\infty$. The left hand side of (\ref{IR}) is then equal to the number of times the curve $v_r$ encircles the origin. This shows that the winding number of the curve described by $E(t)$ when $t$ runs from $t \to -\infty$ to $t \to \infty$ is equal $N_\mathrm{zeros}$.  \end{proof}

Signals such that $\tilde E(\omega) = 0$ for $\omega < 0$ are called \textit{minimum phase signals} if, beside having no poles of $E(t+i \tau)$ with $\tau <0$, they have also no zeros with $\tau <0$. 
Theorem \ref{Th1} shows that a necessary and sufficient condition for a signal $E(t)$ to be of minimum phase is that the curve described by $E(t)$ when $t$ runs from $t \to -\infty$ to $t \to \infty$ does not encircle the origin. This theorem permits to decide whether a signal is of minimum phase or not by numerically evaluating the values of $E(t)$ for $t$ real only, avoiding the difficult, ill-conditioned \cite{Gallicchio}, analytic continuation of $E(t)$ in the complex plane. Notice that, although non zero, the constant bias $\overline E$ can be arbitrarily small.

\begin{Th} \label{Co10} Let us assume $E_s(t) \in \tilde \mathcal{C}_\beta(0,B)$, and a constant $\bar E \ne 0$, and define $E(t) = E_s(t) + \bar E$ with $\bar E \ne 0$, such that $E(t) \ne 0$, $\forall t \in \mathbb{R}$ and the trajectory of $E(t)$ never encircles the origin for $t \in (-\infty, \infty)$. Then, it is possible to define 
\be G(z) = \log\left[ \frac{E(z)}{\overline{E}} \right], \label{Ft} \ee
with the determination of the logarithm chosen such that a) $\lim_{t \to -\infty} G(t+0i) = 0$ and b) $G(z)$ is a holomorphic function in the lower half plane including the real axis. With this choice, if we define $G(t)$ as the restriction of $G(z)$ on the real axis, $G(t) \in L^1 \bigcap L^2$, and its Fourier transform $\tilde G(\omega)$ is such that $G(\omega) = 0$ for $\omega < 0$. \end{Th}

\begin{proof} The continuation of $E(t) = E_s(t) + \bar E$ in the complex plane, $E(z)$ is an entire function, because $E_s(z)$ is an entire function. It is possible to define a function $G(z)$ holomorphic in the lower half plane including the real axis and such that $E(z) = \bar E \exp[G(z)]$ by the following procedure. Let us construct $G(z)$ for any $z$ in the lower half plane by analytic continuation from an arbitrary point $z_0$ in the lower half plane using
\be  G(z) = G(z_0) + \int_{z_0}^z \frac{\dot E(z')}{E(z')} \df z', \label{Gz1} \ee
where the integration path from $z_0 = $ to $z$ is contained entirely in the lower half plane but otherwise arbitrary. The poles of $\dot{E}(z)/E(z)$ are the poles and the zeros of $E(z)$, $E(z)$ is an entire function and no zeros of $E(z)$ exist in the lower half plane by virtue of theorem \ref{Th1} and on the real axis by assumption, so that $\dot{E}(z)/E(z)$ is a holomorphic function in the lower half plane including the real axis and hence $G(z)$ is also holomorphic in the same domain. The function $G(z)$ is however not unique, because it depends on the choice of $z_0$ and of $G(z_0)$. We will choose $z_0$ on the real axis such that $z_0 = \lim_{t \to -\infty} (t + 0 i)$, and we will take for $G(z_0)$ the principal value of $\log[E(z_0)/\bar E]$, namely
\be G_\infty = \lim_{t \to -\infty} \log\left[1+\frac{E_s(t)}{\bar E}\right] = 0. \label{Gz0} \ee
We have then
\be G(z) = \lim_{t \to -\infty} \int_{t+0i}^z \frac{\dot E(z')}{E(z')} \df z'. \label{Gz} \ee
Let us now consider the integral in the complex plane $z = t + i \tau$ for $\omega = -|\omega| <0$,
\be J_\Gamma = \oint_{\Gamma} G(z) \exp(-i |\omega| z) \df z, \label{JGm} \ee
where $\Gamma$ is again a contour that includes the real axis and returns to $t = -\infty$ from the semicircle in the lower half plane $z = R \exp(i \phi)$ with $\phi \in [\pi, 2 \pi]$. The integral of Eq. (\ref{JGm}) can be decomposed into two terms
\be J_\Gamma = \tilde G(\omega) + J_C \ee
the first,
\be \tilde G(\omega) = \int_{-\infty}^\infty G(t) \exp\left(i \omega t\right) \df t \label{Gw}, \ee
being the Fourier transform of $G(t)$, and the second, $J_C$, being the contribution of the infinite semicircle. The second contribution is equal to $J_C = \lim_{R \to \infty} J_R$, where
\be J_R = i \int_{\pi}^{2 \pi} G[R \exp(i \phi)] \exp\{i [-|\omega| R \exp(i \phi)+ \phi]\} R \df \phi. \ee
In the absence of poles in the lower half-plane we have $J_\Gamma = 0$ and hence
\be \tilde G(\omega) = - J_C \quad \forall \omega <0. \label{Gw0} \ee
We have 
\bea |J_R| &\le& \int_{\pi}^{2 \pi} \left|G[R \exp(i \phi)] \right| \exp[|\omega| R \sin(\phi)] R \df \phi \nonumber \\
&\le& M(R) \int_{\pi}^{2 \pi} \exp[|\omega| R \sin(\phi)] R \df \phi, \eea
where $M(R) = \max_{\phi \in [\pi, 2 \pi]} \left\{\left|G[R \exp(i \phi)] \right|\right\}$. We have then
\be |J_R| \le 2 M(R) \int_{0}^{\pi/2} \exp[-|\omega| R \sin(\phi)] R \df \phi, \ee
where the integral at right-hand side tends to a constant
\be \lim_{R \to \infty} \int_{0}^{\pi/2} \exp[ - |\omega| R \sin(\phi)] R \df \phi = \frac 1 {|\omega|}. \ee
Using now Eq. (\ref{Gz}) on the semicircle of radius $R$ we obtain for $\phi \in [\pi, 2 \pi]$
\bea  G_\phi &=& \lim_{R \to \infty} G[R \exp(i \phi)] \nonumber \\
&=& \lim_{R \to \infty} \log\left\{1 + \frac{E_s[R\exp(i \phi)]}{\bar E} \right\} = 0, \label{GR} \eea
because the principal value of the logarithm is used in Eq. (\ref{Gz0}) and, being $\lim_{t \to -\infty} G(t+0i) = G_{\phi = \pi} = 0$, the principal value should be used, for continuity, in $G_\phi$ for all $\phi \in (\pi, 2 \pi]$. Consequently we also have $\lim_{R \to \infty} M(R) = \max_{\phi \in [\pi, 2 \pi]} G_\phi = 0$, and hence $J_C =  \lim_{R \to \infty} J_R = 0$. 

The above results implies that also $\lim_{t \to \infty} G(t+0i) = G_{\phi = 2 \pi} = 0$, so that asymptotically, on the real axis
\be G(t) = \log\left[1+\frac{E_s(t)}{\bar E} \right] \to \frac{E_s(t)}{\bar E}, \quad |t| \to \infty. \ee
Being $E(t) > \bar E$, $\forall t \in \mathbb R$, and $E(t)$ finite, the integrals of both $|G(t)| $ and $|G(t)|^2$ over a finite interval are finite. At infinity, being $E_s(t) \in \tilde \mathcal{C}_\beta(0,B)$ for $\beta >0$, Eqs. (\ref{mod1}) and (\ref{Hb}) show that $|E_s(t)| \simeq 1/|t|^2$ for $|t| \to \infty$ and hence the integrals of $|G(t)| $ and $|G(t)|^2$ also converge at infinity. Consequently, $G(t) \in L^1 \bigcap L^2$. \end{proof}

The above theorem shows that the ambiguity in the determination of the logarithm to be used in the restriction of $G(z)$ to the real axis, $G(t) = \log[E(t)/\bar E]$, is removed by the prescriptions that both $\lim_{t \to -\infty} G(t) = 0$ and $G(t)$ is a continuos function of $t \in \mathbb R$.

\begin{Th} \label{Th3} Let us assume $E_s(t) \in \tilde \mathcal{C}_\beta(0,B)$, and a constant $\bar E \ne 0$, and define $E(t) = E_s(t) + \bar E$ with $\bar E \ne 0$, such that $E(t) \ne 0$, $\forall t \in \mathbb{R}$ and the trajectory of $E(t)$ never encircles the origin for $t \in (-\infty, \infty)$. Then, the phase of $E(t)$ can be reconstructed by a logarithmic Hilbert transform
\be \phi(t) = \overline{\phi} + \frac{1}{2 \pi} \mathrm{p.v.} \int_{-\infty}^{\infty} \df t' \frac{\log\left[ |E(t')|^2 \right]}{t'-t}, \label{phi} \ee
where $E(t) = |E(t)| \exp[i \phi(t)]$ and $\overline{E} = |\overline{E}| \exp(i \overline{\phi})$. \end{Th}

\begin{proof} Under the hypotheses stated, theorem \ref{Co10} ensures that $G(t) \in L^1 \bigcap L^2$ and $\tilde G(\omega) = 0$ for $\omega < 0$ and, hence, the hypotheses of lemma \ref{Le1} are verified, so that 
\be G(t) = \frac i \pi \, \mathrm{p.v.}\int_{-\infty}^\infty \df t' \frac{G(t')}{t'-t}. \label{kk1} \ee
Equation (\ref{phi}) is then readily obtained by equating the imaginary parts of both sides of Eq. (\ref{kk1}), using $G(t) = \log|E(t)| - \log \overline E + i [\phi(t)-\overline \phi]$ and that 
\be \frac{1}{\pi} \mathrm{p.v.} \int_{-\infty}^{\infty} \df t' \frac{\log |\overline{E}|}{t'-t} = 0. \ee
\end{proof} 

Notice that, being
\be \frac{1}{2 \pi} \mathrm{p.v.}  \int_{-\infty}^{\infty} \df t \int_{-\infty}^{\infty} \df t' \frac{\log\left[ |E(t')|^2 \right]}{t'-t} = 0, \ee
the phase bias $\overline \phi$ is also the time average of $\phi(t)$ 
\be \overline \phi = \lim_{T \to \infty} \frac 1 T \int_{-T/2}^{T/2} \df t \phi(t). \ee

From the necessary and sufficient condition of theorem \ref{Th3} more restrictive sufficient conditions can be derived. One is the following

\begin{Co} \label{re} Let us assume $E_s(t) \in \tilde \mathcal{C}_\beta(0,B)$, and define $E(t) = E_s(t) + \bar E$, with $\bar E \ne 0$ a complex constant. Then, $E(t)$ is of minimum phase, and hence its phase can be reconstructed by the logarithmic Hilbert transform (\ref{phi}) when $\exists \phi_0 \in [0, 2 \pi)$ such that $\mathrm{real}[E(t) \exp(i \phi_0)] > 0$, $\forall t \in \mathbb R$.  \end{Co}

\begin{proof} Under the hypotheses of the theorem, the curve $E(t)$ never encircles the origin, hence the hypotheses of theorem \ref{Th3} are satisfied.  \end{proof} 

This corollary can also be proven independently, under slightly less restrictive conditions.

\begin{Th} \label{re1} Let $E_s(t) \in L^1 \bigcap L^2$ such that $\tilde E_s(\omega) = 0$, $\forall \omega < 0$, and define $E(t) = E_s(t) + \bar E$, with $\bar E \ne 0$ a complex constant. Then, the phase of $E(t)$ can be reconstructed by the logarithmic Hilbert transform (\ref{phi}) when $\exists \phi_0 \in [0, 2 \pi)$ such that $\mathrm{real}[E(t) \exp(i \phi_0)] > 0$, $\forall t \in \mathbb R$.  \end{Th}

\begin{proof} Let us define $F(t) = [E(t)/\overline E-Z_0]/Z_0$ with $Z_0 = R \exp(i  \phi_0)$, and use this definition into the $G(t)$ given by Eq. (\ref{Ft}), to obtain $G(t) = \log(Z_0) + \log\left[1 + F(t) \right]$, where for the logarithm we assume its principal value. We then note that the expansion
\be G(t) = \log(Z_0) + \sum_{n=1}^\infty \frac{(-1)^{n+1}}n F^n(t), \label{Z_0} \ee
is legitimate for $|F(t)|<1$, that is inside the circle $|E(t)/\overline E-Z_0|< |Z_0|$ centered in $Z_0$ and of radius $|Z_0|=R$. In the limit $R \to \infty$, the expansion is then legitimate for $F(t)$ inside the half plane delimited by a straight line passing through the origin and orthogonal to $Z_0$ and containing $Z_0$, which is the region $F(t)$ belongs to by virtue of the condition $\mathrm{real}[E(t) \exp(i \phi_0)] > 0$, $\forall t \in \mathbb R$.  Being the spectrum of $E(t)$ zero for $\omega < 0$, the spectrum of $F(t)$ and of all its powers $F^n(t)$ is also zero for $\omega < 0$. The series expansion (\ref{Z_0}) then shows that the spectrum of $G(t)$ exists and it is zero for $\omega < 0$. If we chose the determination of the logarithm such that $G(t) = \log[E(t)/\bar E] \to 0$ as $|t| \to \infty$, then being $E(t) \ne 0, \forall t \in \mathbb R$ and being $G(t) \simeq E_s(t)/\bar E$ for $|t| \to \infty$ with $E_s(t) \in L^1 \bigcap L^2$, also $G(t) \in L^1 \bigcap L^2$. The function $G(t)$ then fulfills once again the conditions of lemma \ref{Le1} so that its real and imaginary parts are the Hilbert transform of one another, and from this the thesis is deduced. \end{proof}

Another sufficient condition, more restrictive than that given by corollary \ref{re} is that a signal $E(t) = \bar E + E_s(t)$ is of minimum phase if $|E_s(t)|^2 < |\overline{E}|^2$ for every $t$. This condition was explicitly stated in \cite{Burge} and \cite{Taylor} and, later, independently rediscovered by others, see for instance \cite{Cassioli} where the condition was given in the context of wireless channel characterizations, and \cite{Mecozzi} where the condition was applied to optical measurements. It is immediate to show that, if this condition is satisfied, $E(t)$ never encircles the origin and it is consequently of minimum phase. In addition, being
\be |E_s(t)|^2 \le \int \frac{\df \omega}{2 \pi} |\tilde E_s(\omega)|^2. \ee
condition $|\overline{E}|^2 > |E_s(t)|^2$ also includes the even more restrictive one
\be |\overline{E}|^2 > \int \frac{\df \omega}{2 \pi} |\tilde E_s(\omega)|^2 \ee 
given in \cite{Kino}.
 
Although we assumed $\overline E \ne 0$, the results obtained are valid for $\overline E$ arbitrarily small. The case $\overline E = 0$ is a delicate one, because in this case, $G(t) \notin L^2$, so that $\tilde G(\omega)$ is not defined.  In this case, however, the role of $\overline E$ is not essential for convergence of the logarithmic Hilbert transform (\ref{phi}). Indeed, the phase of every $E(t) \in \tilde \mathcal{C}_\beta(0,B)$ is well defined for every $\overline E$, and it does not diverge for $\overline E = 0$. For $\overline E \ne 0$, if $E(t)$ does not encircle the origin, $\phi(t)$ given by Eq. (\ref{phi}) gives the phase of $E(t)$. Consequently, if $E(t)$ does not encircle the origin the phase of $E(t)$ can be still calculated for $\overline E = 0$ using Eq. (\ref{phi}) in the limit $\overline E \to 0$, which for what said does not diverge. We may therefore remove in the following the condition $\overline E = 0$ assuming that this case is included as the limit for $\overline E \to 0$.

\begin{Co}  \label{Cou} Let the two fields $E_0(t) \in \tilde \mathcal{C}_\beta(0,B)$ and $E_0'(t) \in \tilde \mathcal{C}_\beta(0,B)$ having the same intensity $|E_0(t)|^2 = |E_0'(t)|^2$. If $E_0(t)$ and $E_0'(t)$ do not encircle the origin, then they are both minimum phase signals and hence $E_0(t) = E_0'(t) \exp(i \overline{\phi})$, $\forall t$, where $\overline{\phi}$ is an arbitrary constant phase. \end{Co}


An interesting and not immediately obvious property is now the following:

\begin{Th} \label{Th0} Given an arbitrary field $E(t) \in \tilde \mathcal{C}_\beta(0,B)$ then
\be E_0(t) = |E(t)| \exp[i \phi_{0}(t)], \label{mp} \ee 
where
\be \phi_{0}(t)= \overline{\phi} + \frac{1}{2 \pi} \mathrm{p.v.} \int_{-\infty}^{\infty} \df t' \frac{\log\left[ |E(t')|^2 \right]}{t'-t}, \label{mp1} \ee
with $\overline{\phi}$ an arbitrary phase, is band-limited to the same interval of $E(t)$, namely $(0, 2 \pi B)$. \end{Th}

\begin{proof} This theorem is parallel to an analogous one in \cite{Hofstetter} that gives the condition for a function to be the auto-convolution of a time limited function. If the analytic continuation of $E(t)$, namely  $E(t+i\tau)$, has no zeros for $\tau <0$, then we have $E(t) = E_0(t)$ for a suitable value of $\overline{\phi}$. If this is not the case, any zero of $E(z)$ in the lower complex half plane, say $z_0 = t_0 - i |\tau_0|$, can be removed by multiplication by the pure phase modulation
\be H(t) = \frac{t-t_0 - i |\tau_0|}{ t-t_0 + i |\tau_0|} = 1 - \frac{ 2 }{1-i(t-t_0)/|\tau_0|}, \ee
which adds a zero in the upper half plane symmetrically placed with respect to the real axis. The spectrum of the field $E'(t)$ after the phase modulation is
\bea \tilde E'(\omega) &=& \tilde E(\omega) -2 |\tau_0| \int \frac{\df \omega'}{2 \pi} \tilde E(\omega') u(\omega-\omega') \nonumber \\
&& \times \exp\left[ i (\omega-\omega') (t_0-i |\tau_0|) \right], \label{Ep} \eea
where $u(\cdot)$ is the unit step function. The integral at right hand side of (\ref{Ep}) is zero for $\omega < 0$ because the non-zero regions of $E(\omega')$ and $u(\omega-\omega')$ do not overlap, and it is also zero for $\omega>2 \pi B$ because in this case $u(\omega-\omega')$ can be replaced by $1$ and
\be \int \frac{\df \omega'}{2 \pi} \tilde E(\omega') \exp\left[ -i \omega' (t_0-i |\tau_0|) \right] = E(t_0-i |\tau_0|) = 0. \ee
After the phase modulation, the spectrum is still zero for $\omega < 0$ and $\omega > 2 \pi B$, and the zero at $t_0 - i |\tau_0|$ is replaced by a zero at $t_0 + i |\tau_0|$. This procedure can be repeated until all the zeros in the lower complex half plane are removed. At the end, the resulting field is a minimum phase signal still bandwidth limited to the same bandwidth of the original signal. Being the minimum phase signal unique, this signal is equal to $E_0(t)$ given by Eq. (\ref{mp}).  \end{proof} 

Let us now consider the class of functions $\tilde \mathcal{C}(0,B)$ band-limited in the interval $(0,B)$, obtained from $\tilde \mathcal{C}_\beta(0,B)$ in the limit of $\beta \to 0^+$. Theorem \ref{Th0} insures that the class of band-limited signals with a common intensity profile always includes the minimum phase signal. To be more precise, theorem \ref{Th0} insures that if one groups the functions $E(t) \in \tilde \mathcal{C} (0,B)$ into functions with the same intensity profile $I(t) = |E(t)|^2$, any of these classes always includes the minimum phase function $E_0(t) \in \tilde \mathcal{C} (0,B)$. Corollary \ref{Cou} than shows that if the minimum phase condition is met, then the intensity profile uniquely determine the function $E_0(t)$ with the exception of an immaterial rotation of the complex plane. From now on, we will refer for convenience to $E(t)$ as the equivalence class of functions that differ from one another by an arbitrary constant phase factor, and with this caveat corollary \ref{Cou} states that the minimum phase signal is unique. Consequently, all possible intensity profiles $I(t) = |E(t)|^2$ with $E(t) \in \tilde \mathcal{C} (0,B)$ can be set into a one-to-one correspondence with the class of minimum phase functions $E_0(t)$.

From theorem \ref{Th0} we may also draw the important consequence that the ratio between a signal $E(t) \in \tilde \mathcal{C}(0,B)$ that encircles $N$ times the origin and the minimum phase signal $E_0(t)$ with the same intensity profile $|E_0(t)|^2 = |E(t)|^2$, which we define as
\be H_N(t) = \frac{E(t)}{E_0(t)} , \ee
has the form
\be H_N(t) = \prod_{k=1}^N \left[1 - \frac{ 2 }{1-i(t-t_k)/|\tau_k|} \right], \label{HN} \ee 
where $t_k - i |\tau_k|$ are the $N$ zeros of $E(t+i \tau)$ in the lower complex half plane. The function $H_N(t)$ is the product of terms which are 1 minus a Lorentzian line-shape, centered on the real parts $t_k$ and of width equal to the modulus of the imaginary part $|\tau_k|$ of each zero in the lower complex half plane. When these Lorentzian line-shapes are well separated, which may occur when the number $N$ is small, we have
\be H_N(t) \simeq 1 - \sum_{k=1}^N \frac{2}{1-i(t-t_k)/|\tau_k|}, \label{HNapp} \ee 
so that $H_N(t)$ is mostly 1 with the exception of small regions of amplitude $|\tau_k|$ around the time $t_k$ where the deviation from unit has a Lorentzian shape. The magnitude of the imaginary parts of the zeros in the lower complex half plane $|\tau_k|$ is proportional to the amplitude of the region around $t_k$ where the reconstruction of the phase with the logarithmic Hilbert transform is inaccurate. 

\section{Encoding information on the intensity of an electromagnetic field} 

Let us now consider the following problem: What is the most efficient transmission over a bandwidth $B=1/T$ for the complex field $E(t)$ if we have the capability of modulating the optical field in modulus and phase and we perform square-law intensity detection of the transmitted signal at the receiver? The study of transmission systems where only direct detection is applied at the receiver became lately an intense area of research, targeting high capacity applications in the short reach range \cite{Randel_invited,DMT,Lowery,Lowery1,Petermann,Randel}. Although the information is not contained in the phase of the optical field, a suitable phase shaping is required to confine the signal within a limited bandwidth. One could in principle think of generating an arbitrary amplitude modulation at bandwidth $2B$ and then use a phase modulation equal to the logarithmic Hilbert transform of the intensity that makes this signal single sideband, as proposed in \cite{Petermann}, but in general the spectrum of the signal after modulation is not zero for $\omega > 2 \pi B$, and is in principle unbounded for positive frequencies. 

The analysis of the present paper gives a general solution to this problem. We have shown that all possible field patterns of a given optical bandwidth can be grouped into classes of fields with the same intensity pattern, and theorem \ref{Th0} ensured that to each of these classes belongs the minimum phase signal $E_0$. Restricting the set of symbols to minimum phase signals therefore does not reduce the set of symbols available to transmission. The recently proposed Kramers Kronig (KK) transmission scheme \cite{KK} is capable of receiving without errors any minimum phase signal $E_0(t)$ and hence it maximizes the set of symbols that can be received over a given optical bandwidth. The principle of operation of the receiver is very simple: the receiver detects the intensity profile only, and the phase of the field is calculated from the intensity profile using the logarithmic Hilbert transform (\ref{phi}) over an integration time window much longer than the inverse of the optical field bandwidth. Assuming that chromatic dispersion is optically compensated or pre-compensated at the transmitter, we will show in the section that follows that a signal constellation like (\ref{mod}) with real$(a_n) > a$ and imag$(a_n) > a$ with $a$ a suitably chosen positive constant  produces a minimum phase signal. It is enough, in general, to choose $a$ much smaller than the range of possible values of $a_n$, so that the available symbols are almost all those belonging to one quadrant, say the first, of the complex plane, one quarter of all possible symbols available on the two quadratures. If some additive noise impairs the signal field, and if the noise is small enough that the perturbed field does not encircle the origin, then the perturbed field is still of minimum phase and faithfully detected by the receiver. Therefore, in spite of the fact that the quadratic receiver detects the field intensity only, and that the phase is calculated from the intensity profile and not independently measured, the receiver acts as a linear receiver (in the field) that is capable of detecting signals in the first quadrant only. Consequently, for high signal-to noise ratio, we may conjecture that the capacity of such system is of the order of the capacity (per unit bandwidth) of a coherent system in which the symbols are constrained in the first quadrant only, approximately two bit less than the capacity of a full, unconstrained, coherent system \cite{KK,CapacityIMDD}.

\section{Numerical validation} \label{num}

We numerically validated the results of the previous sections by testing on a signal made as the sequence of $512$ waveforms of the form $(\ref{mod})$ with $a_n = [(b + k_1) + i(b + k_2)] \sqrt{T}$, where $n$ runs from $1$ to $512$ and $k_1$ and $k_2$ are randomly chosen between $0$ and $7$. This signal corresponds to a shifted 16 quadrature-amplitude modulation (16QAM) constellation. The bias $b$ was chosen equal to two values, one $b_3 = 0.5$ insufficient to make the waveform of minimum phase, and the second $b_0 = 1.1$, that makes the waveform a minimum phase one. 
\begin{figure}[!t]
\centering
\includegraphics[width=2.5in]{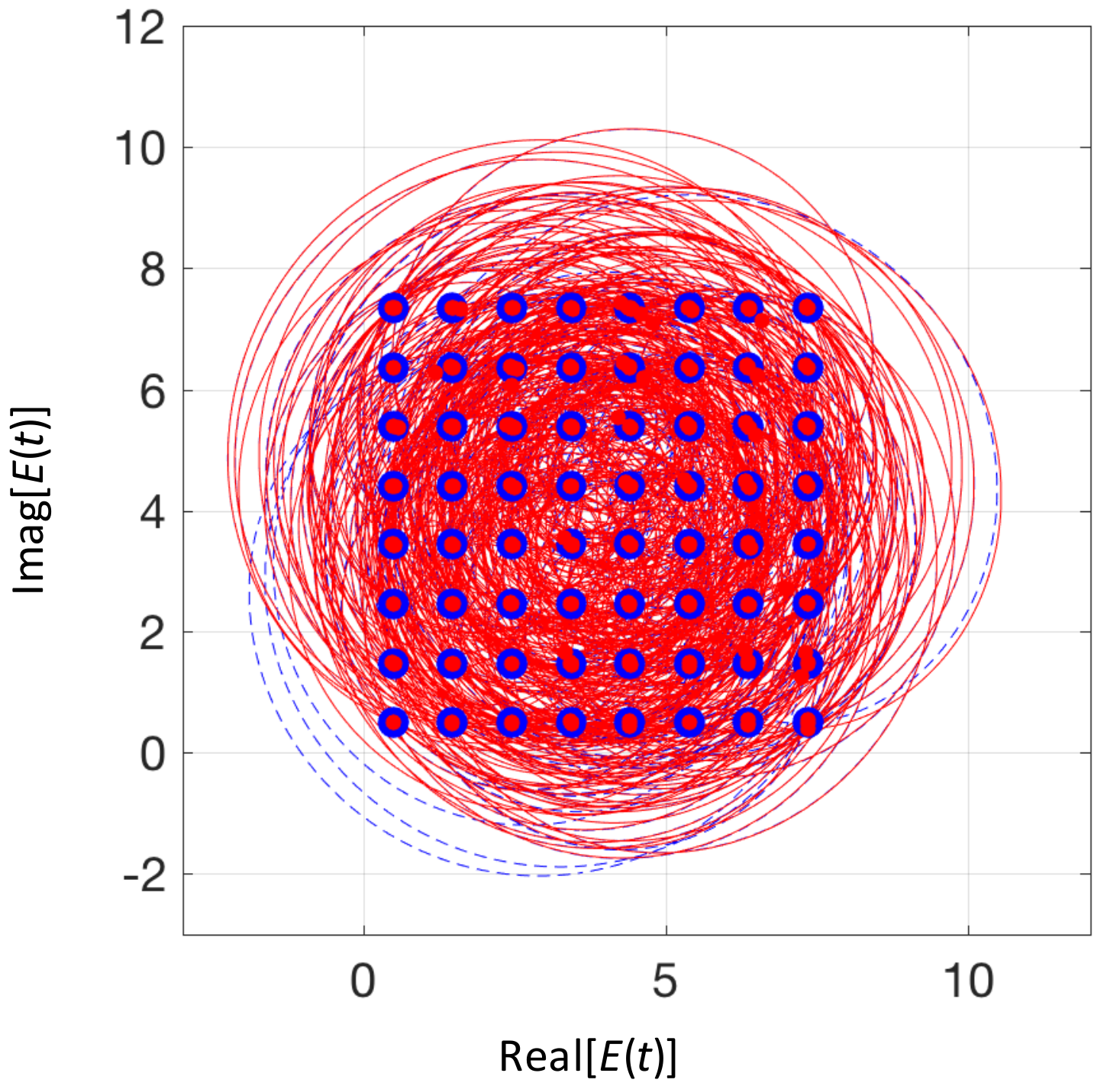}
\caption{Complex field of a shifted 16QAM. Red solid line: reconstructed field $E_0(t)$, blue dashed line, detected field $E(t)$, for a bias $b_3 = 0.5$.}
\label{fig1}
\end{figure}

The numerical analysis was performed using a Marlab program, performing the Hilbert transform over a finite time window $T_w$ that in the examples given was chosen as $T_w = 512 \, T$, by multiplication of the spectrum of the signal by $-i \, \mathrm{sgn}(\omega)$ and using a fast Fourier transform (FFT) and inverse FFT routines, i.e. implicitly assuming a periodic signal instead of an $L^2$ one. The analysis of the periodic case implied by the use of the FFT can be done by replacing the kernel of the Hilbert transform $1/t$ by its $T_w$-periodic counterpart
\be \sum_{k=-\infty}^\infty \frac 1 {t - k T_w} = \frac {\pi} {T_w} \mathrm{cot}\left(\frac{\pi t}{T_w} \right), \quad t \ne \tau, \label{cot} \ee
and limiting the integration over $\tau$ to a single period $[-T_w/2, T_w/2)$. This corresponds to replace Fourier integrals with Fourier series. All the theorems that we have discussed retain their validity. Theorem \ref{Th1} for instance can be shown without closing the curve $\Gamma$ with a semicircle $C$ at infinity, because when $t$ spans on the real axis an entire period the function $E(t)$ describes in the complex plane a closed curve. While the rigorous analysis of the periodic case is beyond the scope of this paper, the correspondence between the periodic and the $L^2$ case can be obtained considering the case $T_w \to \infty$. In the $L^2$ case, the field $\overline E$ is the time average of the field $E(t)$ 
\be \overline E = \lim_{T_w \to \infty} \frac 1 {T_w} \int_{-T_w/2}^{T_w/2} E(t) \df t, \ee
so that in the periodic case the role of the field bias $\overline E$ is played by the average of the signal over the length of the symbol sequence $T_w$
\be \overline E_\mathrm{av} = \frac 1 {T_w} \int_{-T_w/2}^{T_w/2} \df t E(t). \label{Eav} \ee
An intuitive understanding of this correspondence can be obtained by looking at the spectrum of $E(t) = E_s(t) + \bar E$, namely $\tilde E(\omega) = \tilde E_s(\omega) + 2 \pi \bar E \delta(\omega)$, and the spectrum of the periodic sequence. Both $\overline E$ and $\overline E_\mathrm{av}$ are the amplitude of the spectral component at zero frequency, the Dirac delta function in $\tilde E(\omega)$ being replaced by the amplitude of the discrete spectral component at zero frequency in the periodic case.

Figure \ref{fig1} shows by a blue dashed line a field $E(t)$ obtained with the lowest value of the bias, and by a red solid line the curve $E_0(t)$ reconstructed by the logarithmic Hilbert transform. Big blue dots represent the values of the $E(t)$ at $t = nT_w$, and red smaller dots the values of the reconstructed field at the same times. The curve $E(t)$ has a winding number around the origin of 3. The red solid curve overlaps with the blue dashed one almost everywhere, with the exception of the vicinity of the points where the windings occur. The accuracy of the reconstruction is more evident if we plot the phase reconstructed by the logarithmic Hilbert transform on top of the phase of $E(t)$. This comparison is shown in Figs. \ref{fig2} and \ref{fig3}, where we show by a red solid line the reconstructed phase and by a blue dashed line the phase of $E(t)$. In Fig. \ref{fig2} three phase jumps corresponding to the three windings of $E(t)$ around the origin are clearly visible. Figure \ref{fig3} is the zoom of the plot of Fig. \ref{fig2} in the vicinity of the second phase jump, showing that the deviation of the phase reconstruction from the phase of $E(t)$ is confined to a small region around the jump. 

The numerical results confirm the conjecture proposed the previous section that the constraint of square-law intensity detection reduces the number of available symbols by approximately one quarter, implying a capacity reduction of two bits over the full coherent detection for the same optical bandwidth.
\begin{figure}[!t]
\centering
\includegraphics[width=2.5in]{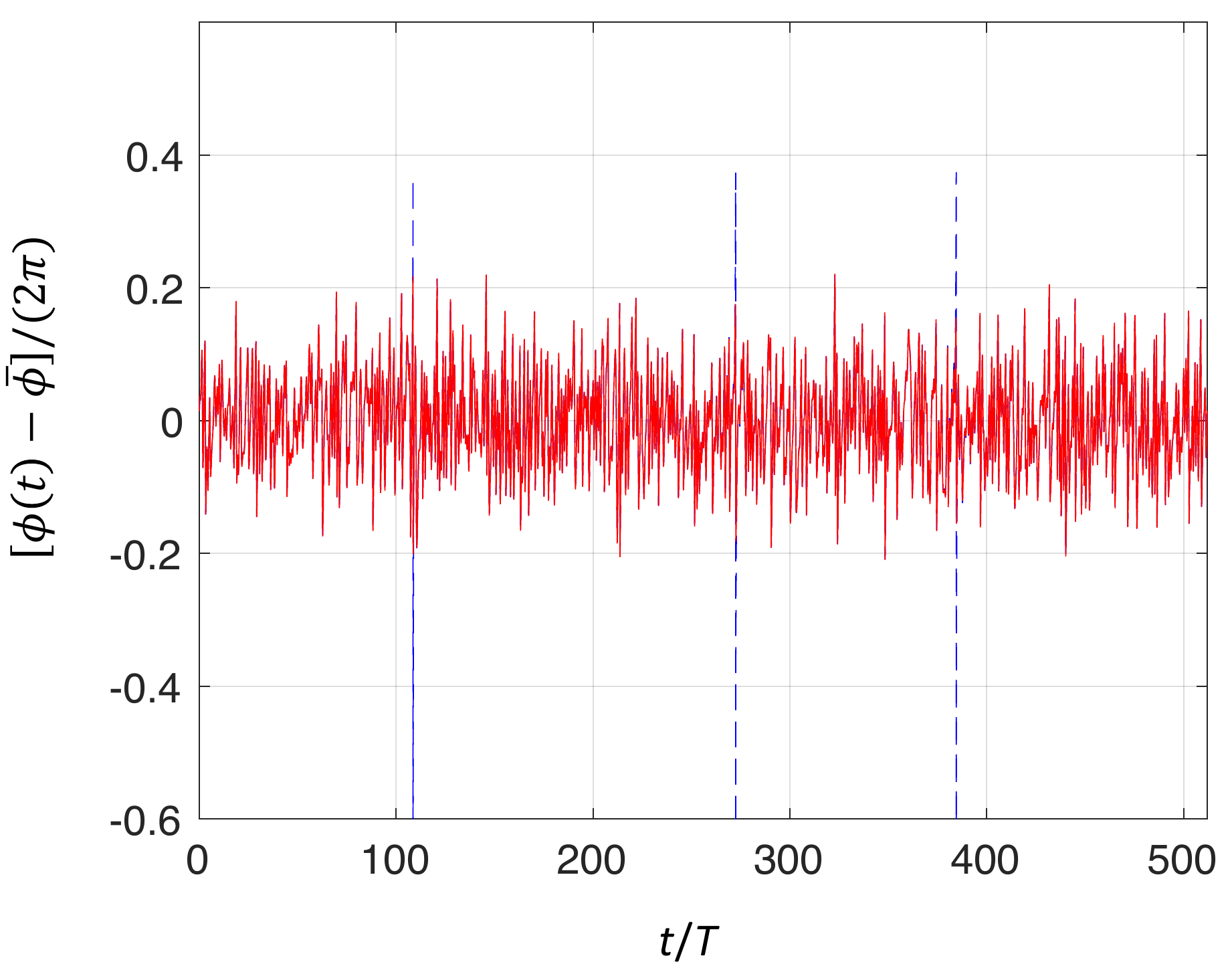}
\caption{Phase of a shifted 16QAM. Red solid line: reconstructed phase profile, blue dashed line, phase of $E(t)$, vs. $t/T$, for a bias $b_3 = 0.5$.}
\label{fig2}
\end{figure}
\begin{figure}[!t]
\centering
\includegraphics[width=2.5in]{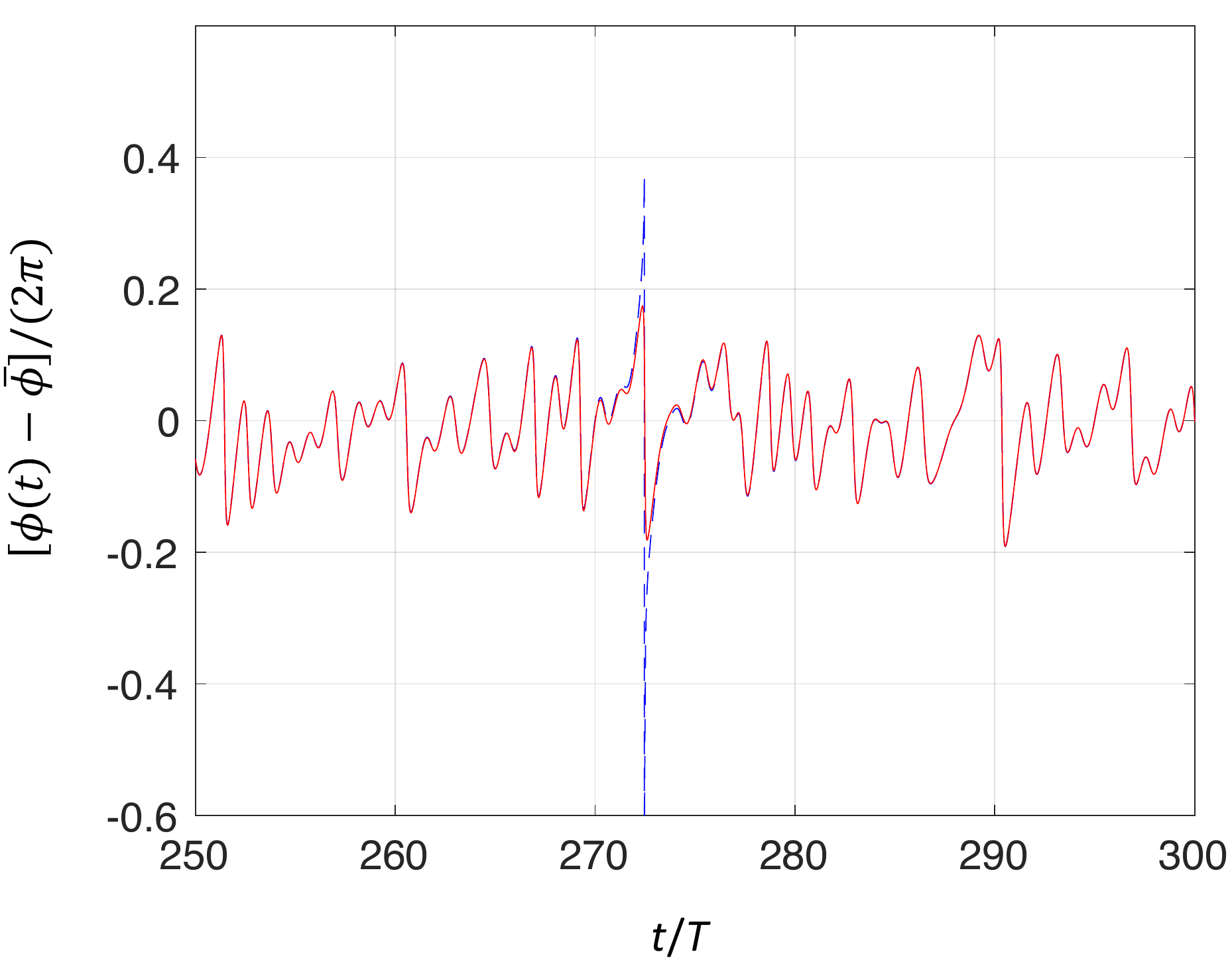}
\caption{Phase of a shifted 16QAM. Red solid line: reconstructed phase profile, blue dashed line, phase of $E(t)$, vs. $t/T$, for a bias $b_3 = 0.5$. This is the zoom in the region of the second phase jump of Fig. \ref{fig2}}
\label{fig3}
\end{figure}
%

%
%
In Figs. \ref{fig4} we show the ratio $|H_3(t) -1| = |E(t)/E_0(t) - 1|$ for the field $E(t)$ of Fig. \ref{fig1}. Figure \ref{fig4} clearly shows the three Lorentzian of amplitude 2 corresponding to the three encircling of the origin of $E(t)$. Figure \ref{fig5} shows by a blue solid thin line the same curve in a semilogarithmic scale, and by a red dashed thick line an interpolation with the curve
\be  H_3(t) - 1 = \prod_{k=1}^3 \left[1 - \sum_{h=-\infty}^{\infty} \frac{ 2 }{1 - i (t-t_k - h T_w)/|\tau_k|}\right], \label{H3} \ee
which, using Eq. (\ref{cot}), becomes
\be  H_3(t) - 1 = \prod_{k=1}^3 \left[1 - \frac{2 \pi i |\tau_k|}{T_w} \mathrm{cot} \left(\pi \frac{t-t_k+i |\tau_k|}{T_w} \right) \right]. \label{H3bis} \ee
Equation (\ref{H3}) was obtained adapting Eq. (\ref{HN}) to account for the temporal periodicity induced by the use of the FFT algorithm, which introduces an infinite number of replicas of the Lorentzian line-shapes spaced by the time window $T_w$. The parameters were obtained by interpolation of the main peaks only and were for the real parts $t_1/T_w=108.5755$, $t_2/T_w = 272.4868$, and $t_3/T_w = 384.5205$, and for the imaginary parts $|\tau_1|/T_w = 0.0070$, $|\tau_2|/T_w = 0.027$ and $|\tau_3|/T_w = 0.041$. As shown in Fig \ref{fig4}, the expression given in Eq. (\ref{H3bis}) was accurate more than six orders of magnitude down to the main peaks. However, a single Lorentzian, i.e. only the dominant term with $h =0$ in Eq. (\ref{H3}), is accurate 2 orders of magnitudes down to the main peak, and is sufficient to exactly reproduce the plot in linear scale shown in Fig. \ref{fig3}. From the average time and the temporal width of each Lorentzian line-shape we were able to compute, with high accuracy, the position of the zeros of $E(t+i\tau)$ in the lower complex half-plane without the need of numerical analytic continuation of $E(t)$.

\begin{figure}[!t]
\centering
\includegraphics[width=2.5in]{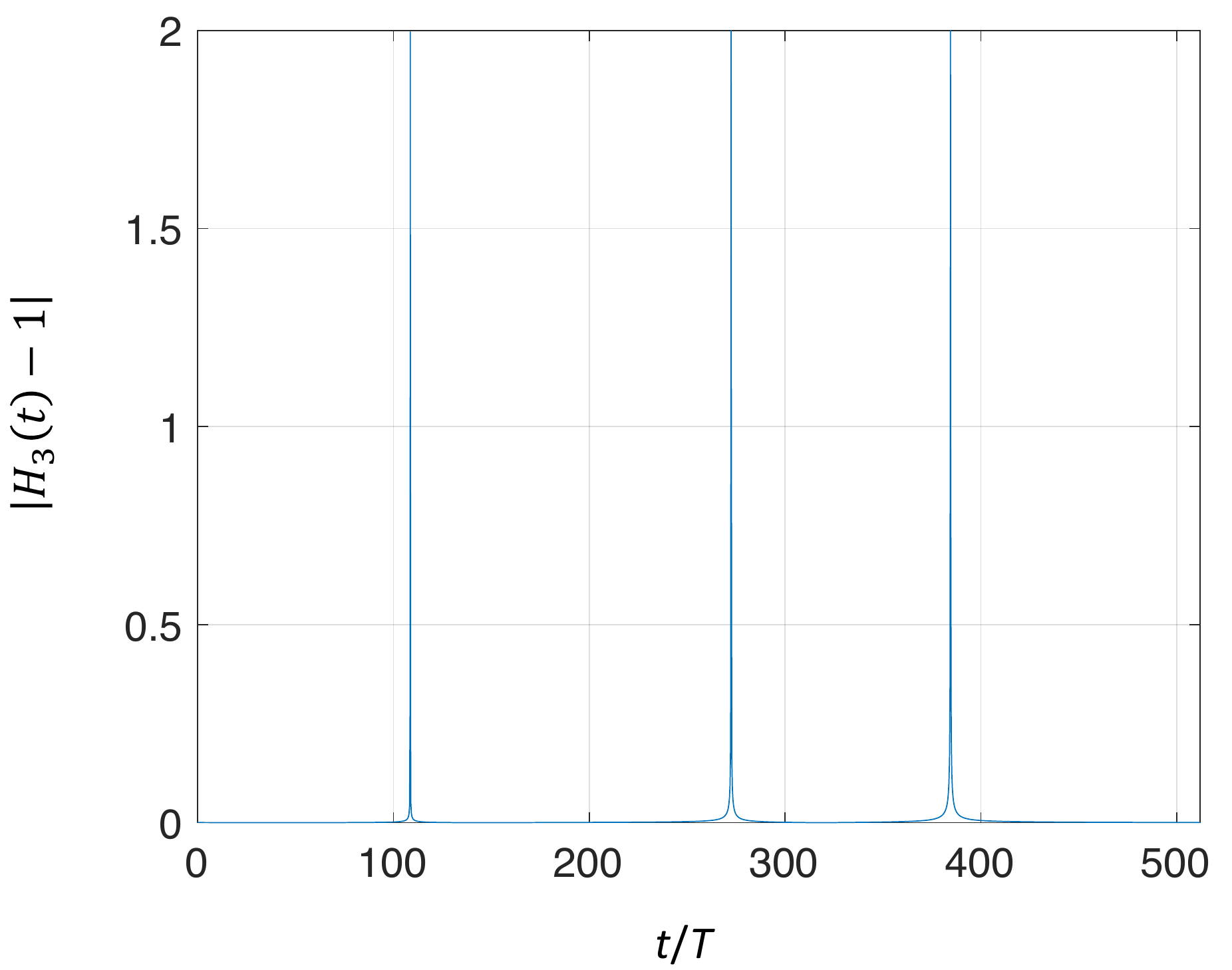}
\caption{$|H_3(t)-1|$ vs. $t/T$ for a bias $b_3 = 1.1$.}
\label{fig4}
\end{figure}
\begin{figure}[!t]
\centering
\includegraphics[width=2.5in]{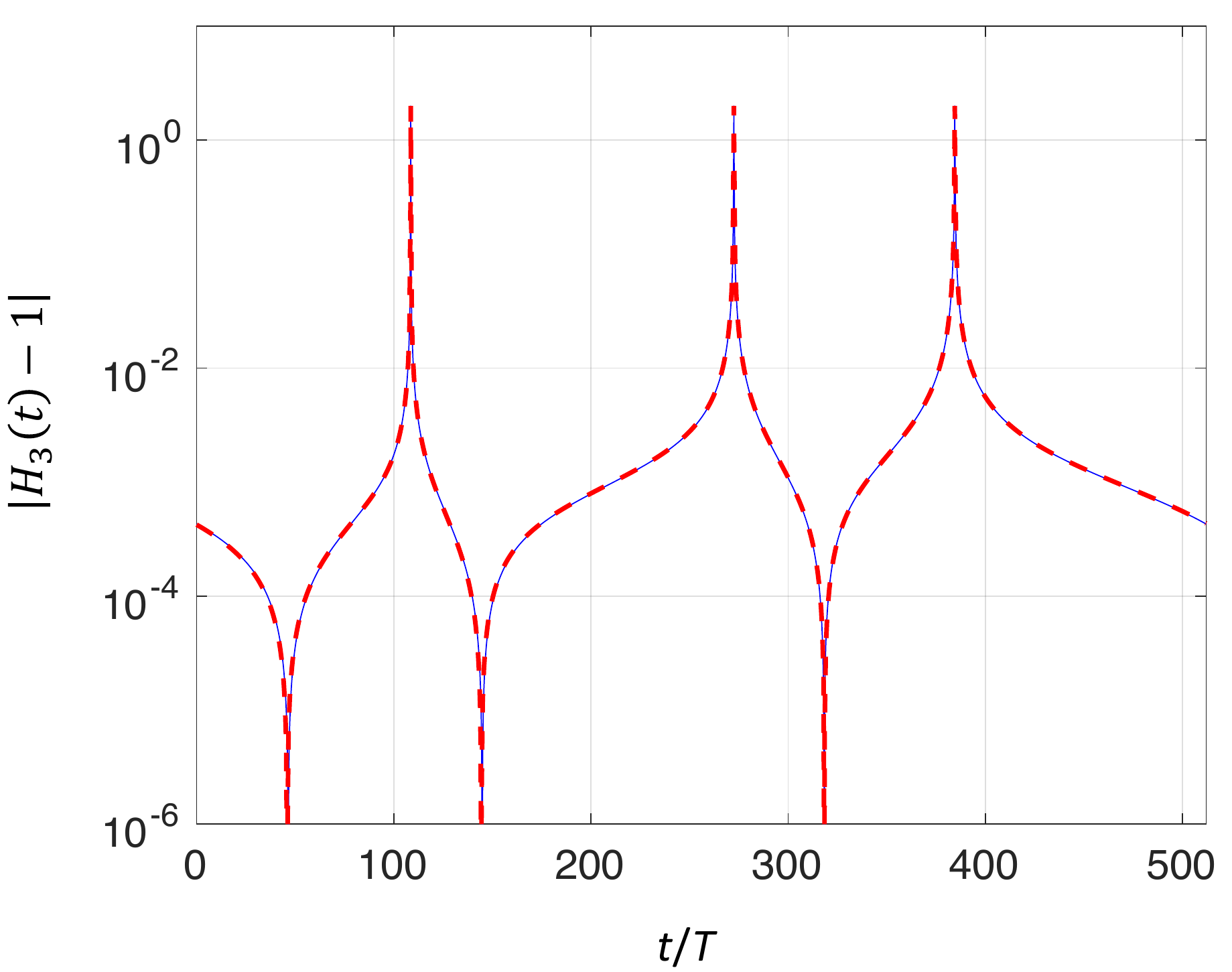}
\caption{$|H_3(t)-1|$ vs. $t/T$ for a bias $b_3 = 1.1$ shown in Fig. \ref{fig4}, in a semilogarithmic scale. The blue solid thin line is the numerical result reported in Fig. \ref{fig4}, the red dashed thick line is the interpolation using Eq. (\ref{H3}).}
\label{fig5}
\end{figure}

Figures \ref{fig6}--\ref{fig8} shows the curves obtained with the same $a_n$ sequence used for Figs. \ref{fig1}--\ref{fig5} but with a larger value of the bias $b_0 = 1.1$. In this case, there are no windings of $E(t)$ around the origin, and the reconstructed field coincides with the original one, $E_0(t) \equiv E(t)$. Figure \ref{fig1} shows the perfect reconstruction of the field at the sampling point by the overlap of the red and the big blue dots.

\begin{figure}[!t]
\centering
\includegraphics[width=2.5in]{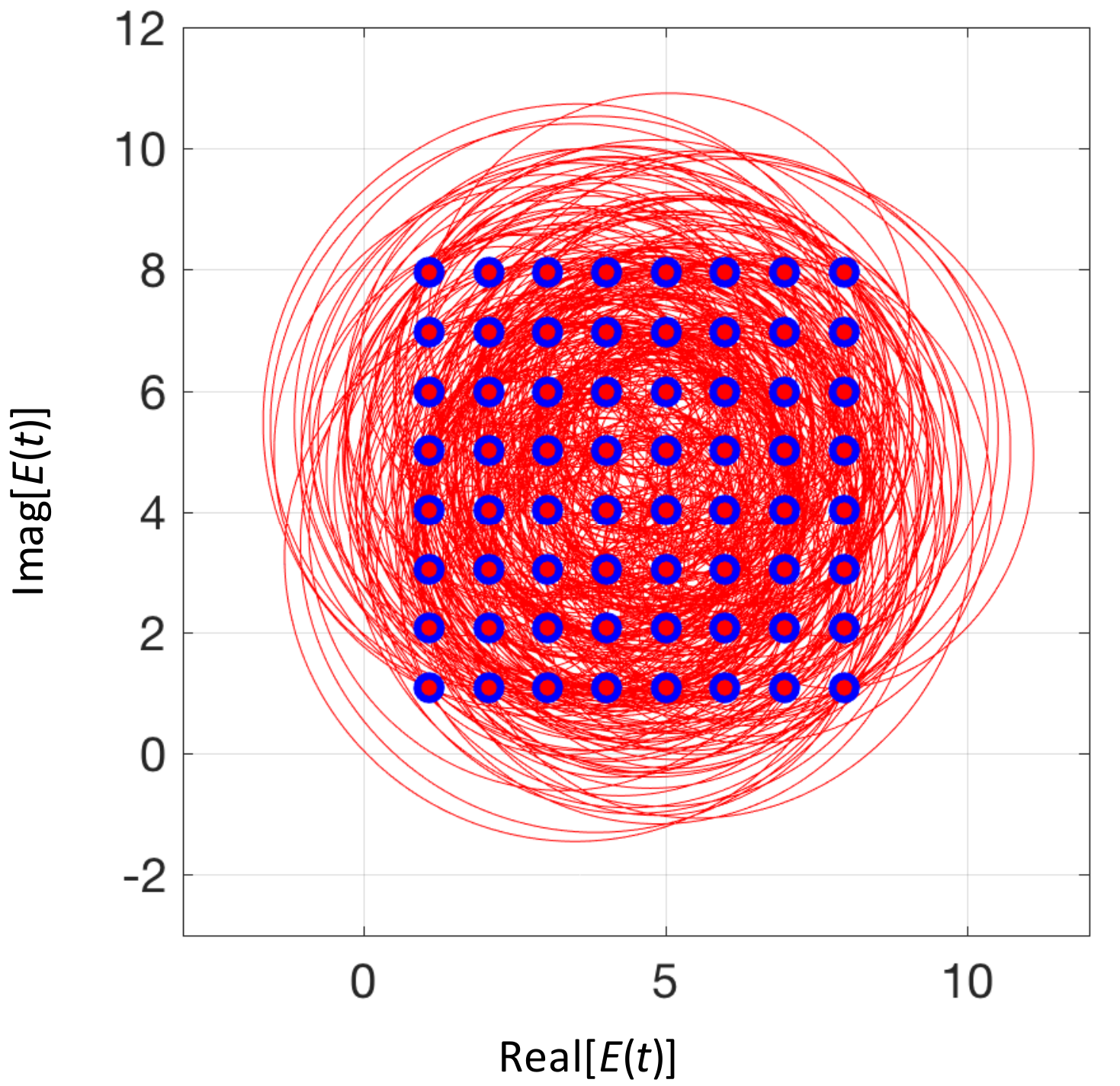}
\caption{Complex field of a shifted 16QAM. Red solid line: reconstructed field $E_0(t)$, blue dashed line, detected field $E(t)$, for a bias $b_3 = 1.1$.}
\label{fig6}
\end{figure}
\begin{figure}[!t]
\centering
\includegraphics[width=2.5in]{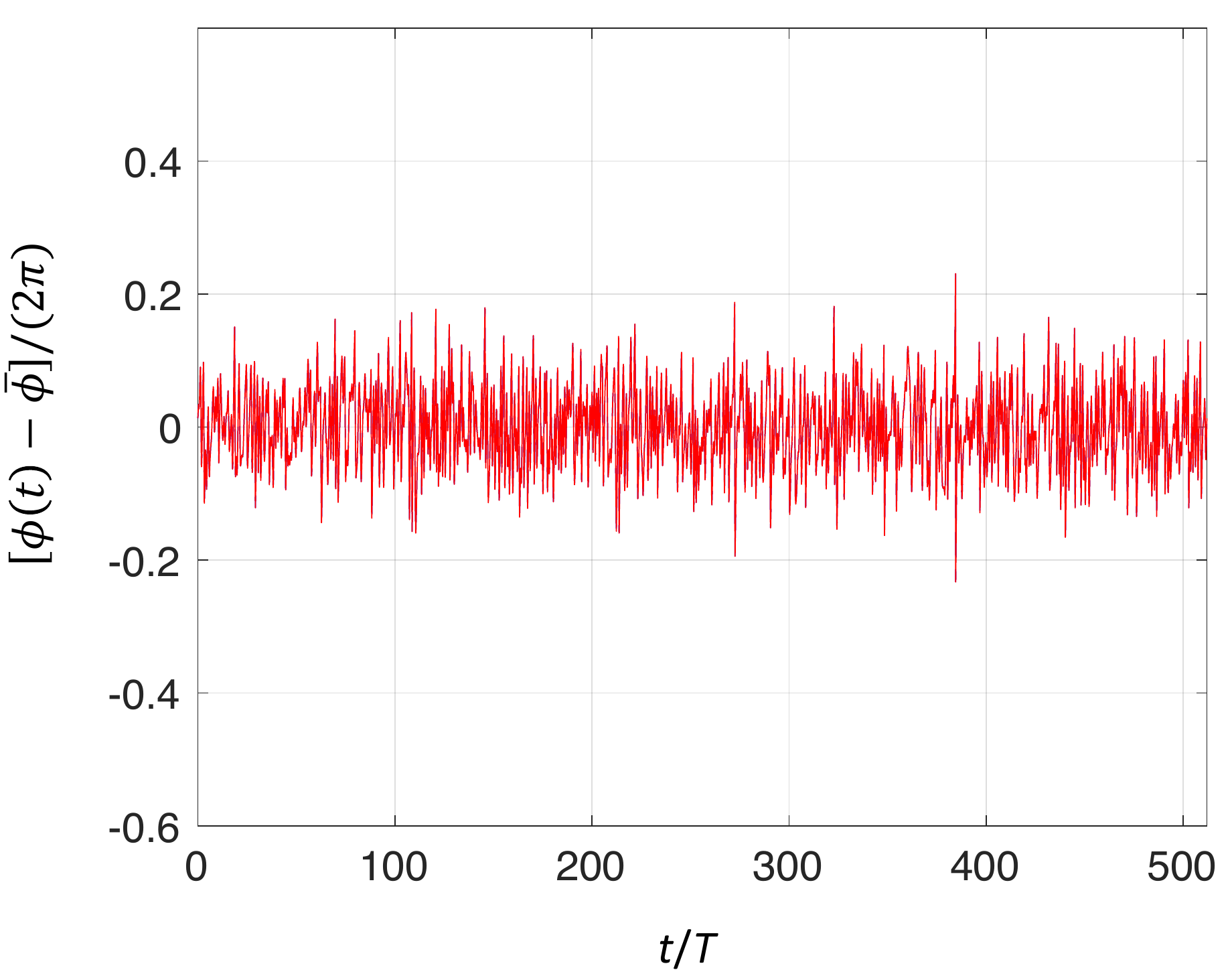}
\caption{Phase of a shifted 16QAM. Red solid line: reconstructed phase profile, blue dashed line, phase of $E(t)$, vs. $t/T$, for a bias $b_3 = 1.1$. The smaller excursion of the phase shown here compared with that of Fig. \ref{fig4} is caused by the larger distance from the origin due to the larger bias.}
\label{fig7}
\end{figure}
\begin{figure}[!t]
\centering
\includegraphics[width=2.5in]{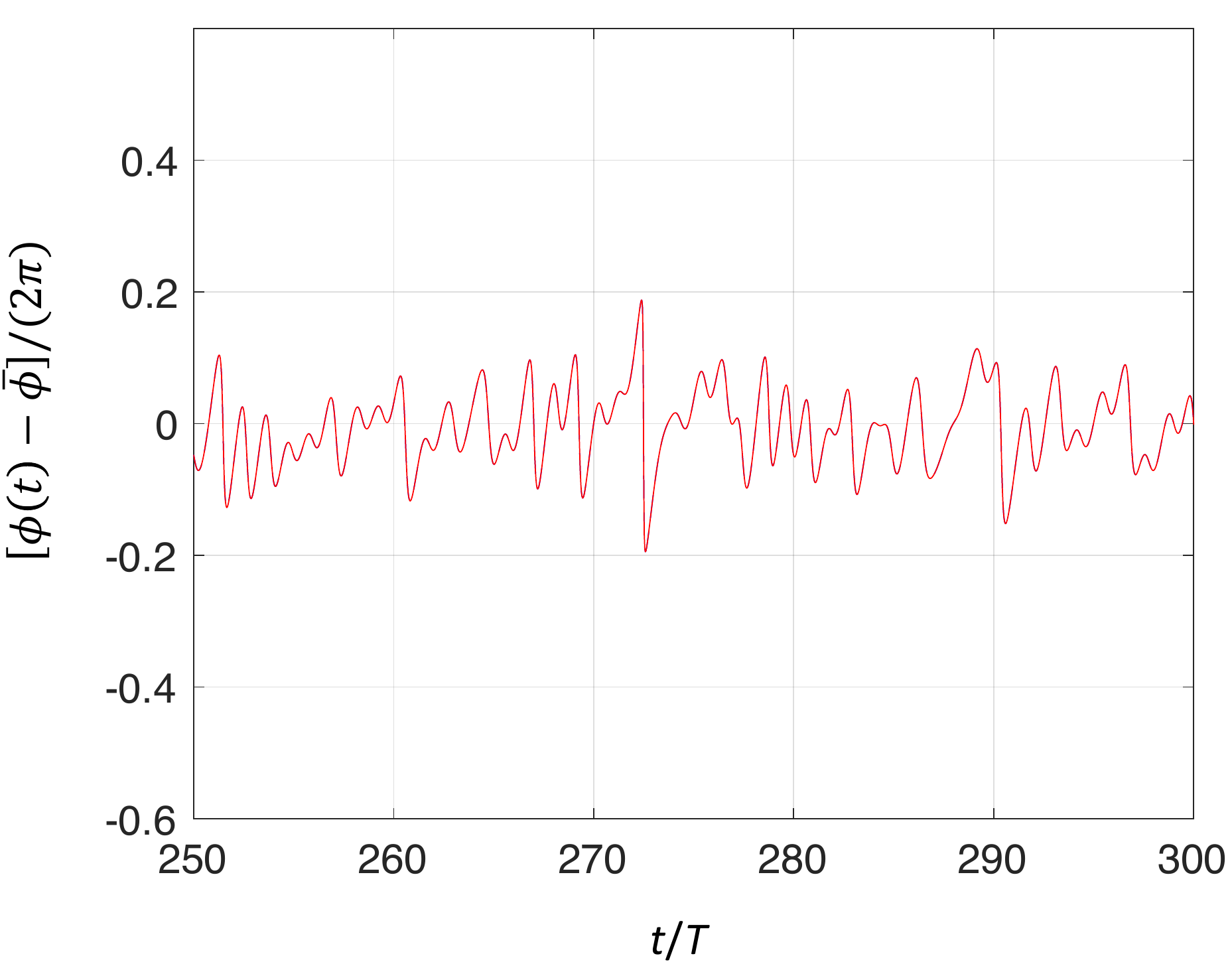}
\caption{Phase of a shifted 16QAM. Red solid line: reconstructed phase profile, blue dashed line, phase of $E(t)$, vs. $t/T$, for a bias $b_3 = 1.1$. This is the zoom in the same region of the phase jump of Fig. \ref{fig5}.}
\label{fig8}
\end{figure}

The example given in Figs. \ref{fig4}--\ref{fig6} corresponds to a case in which beside the condition of no windings around the origin of the trajectories of $E(t)$ also the condition $|E_s(t)|^2 < \overline E^2$ was fulfilled. Figures \ref{fig9} and \ref{fig10} show instead a case where the condition $|E_s(t)|^2 < \overline E^2$ fails but the reconstruction of the phase and consequently of the field from the intensity profile is perfect because no windings around the origin occur. These curves were obtained using $a_n = (b + k) \sqrt{T}$, with $n$ running from $1$ to $512$ and $k$ randomly chosen between $0$ and $7$, corresponding to an amplitude modulation with 8 levels. The value of $b$ was $b_0 = 0.1$.
\begin{figure}[!t]
\centering
\includegraphics[width=2.5in]{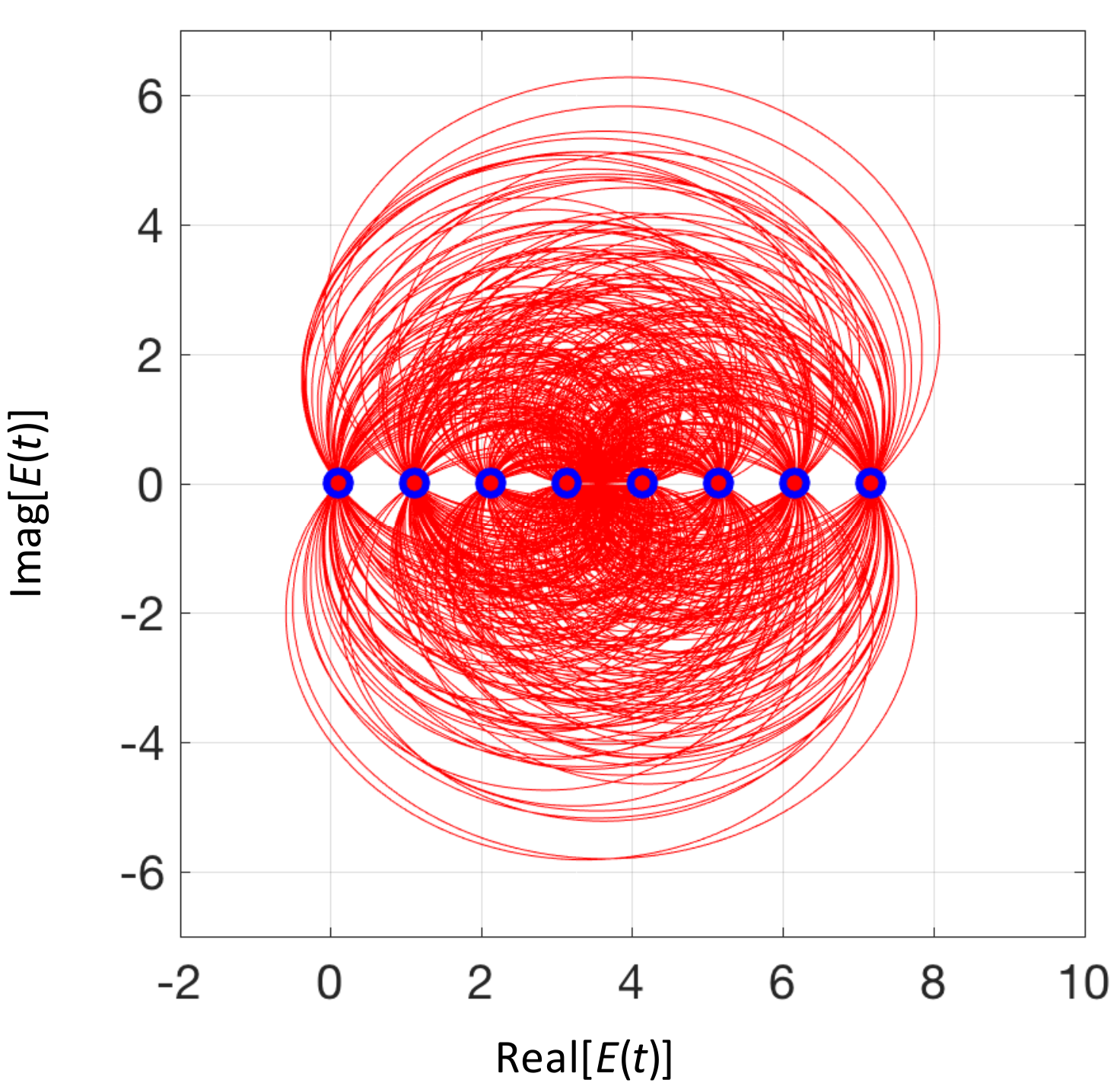}
\caption{Complex field of a shifted 8AM. Red solid line: reconstructed field $E_0(t)$, blue dashed line, detected field $E(t)$, for a bias $b_0 = 0.1$.}
\label{fig9}
\end{figure}
\begin{figure}[!t]
\centering
\includegraphics[width=2.5in]{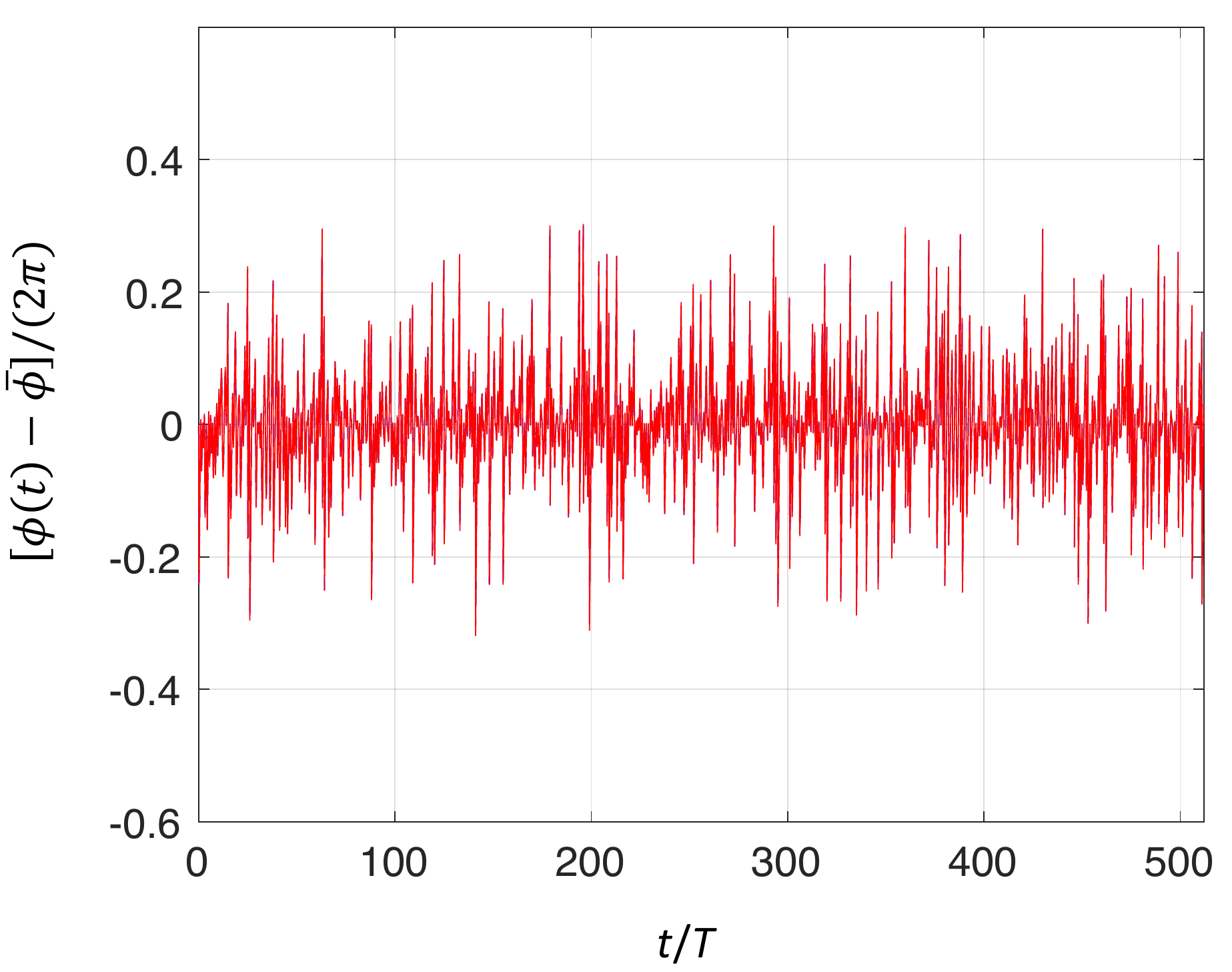}
\caption{Phase of a shifted 8AM. Red solid line: reconstructed phase profile, blue dashed line, phase of $E(t)$, vs. $t/T$, for a bias $b_0 = 0.1$.}
\label{fig10}
\end{figure}

\section{Time-frequency duality} \label{tfd} 

In this paper, we have considered functions $E(t)$ whose Fourier transform $\tilde E(\omega)$ was zero for $\omega < 0$, and our goal was to calculate the phase of $E(t)$ once its intensity $|E(t)|^2$ was measured. Once the necessary changes are made, this case has a one to one correspondence with the case studied in \cite{Fienup,Taylor,Shechtman,Cassioli,Mecozzi} of causal functions $E(t)$, i.e. such that $E(t) = 0$ for $t \le 0$, where the goal was to find the phase of $\tilde E(\omega)$ once the power spectrum $|E(\omega)|^2$ was measured over a range of frequency. The necessary changes include, for instance, the fact that the property $\tilde E(\omega) =0$ for all $\omega < 0$ implied, as we have seen, that the analytic continuation of $E(t)$ in the complex plane does not have poles in the lower complex half plane, whereas the causality of $E(t)$ considered in \cite{Fienup,Taylor,Shechtman,Cassioli,Mecozzi}, namely $E(t) = 0$ for all $t \le 0$, implied that the analytic continuation of $\tilde E(\omega)$ in the complex $\omega$ plane had no poles in the upper complex half plane. The necessary and sufficient condition becomes that $\tilde E(\omega)$ does not encircle the origin when $\omega$ goes from $-\infty$ to $\infty$

\section{Conclusions}

We have given a necessary and sufficient condition for a function $E(t)$ to be of minimum phase, and hence for its phase to be univocally determined by its intensity $|E(t)|^2$. The check of this condition requires only the plot of the function $E(t)$ for $t$ belonging to the real axis, and does not require the analytic continuation of $E(t)$ in the complex plane. We have shown that sufficient conditions previously proposed can be simply derived from this more general one. As an application to communication systems, we find that the recently proposed KK transmission scheme gives a practical way to decode all distinguishable band-limited fields when the detector is sensitive only to the intensity of the field and insensitive to its phase.

%

\end{document}